\documentclass[letterpaper, conference]{IEEEtran}
\usepackage{dsfont,graphics,epsfig,psfrag,stmaryrd,subfigure,amsmath,alltt,amsfonts,amssymb,mathrsfs,txfonts,booktabs,times,algorithm2e,algorithmic,tabularx,ntheorem} \newenvironment{definition}[1][Definition]{\begin{trivlist}
\item[\hskip \labelsep {\bfseries #1}]}{\end{trivlist}}
\IEEEoverridecommandlockouts
\newtheorem{theorem}{theorem}
\newtheorem{remark}{remark}

                            \def\@IEEEmarginE{0.680in}%
                            \def\@IEEEmarginW{0.775in}
\begin{document}
\title{\ \\ \LARGE\bf Topological chaos and chaotic iterations\\Application to Hash functions
\thanks{Christophe Guyeux and Jacques M. Bahi
are with the Computer Science Laboratory LIFC, University of Franche-Comte,
16, route de Gray - 25030 Besan\c con, France (phone: +33 381
666948; email: \{christophe.guyeux, jacques.bahi\}@univ-fcomte.fr).}
}
\author{Christophe Guyeux\\ Jacques M. Bahi, Senior Member IEEE}
\maketitle
\begin{abstract}
% Résumé axé sur l'utilisation du chaos pour la sécurité informatique : théorie et pratique.
This paper introduces a new notion of chaotic algorithms. These algorithms are iterative and are based on so-called chaotic iterations. Contrary to all existing studies on chaotic iterations, we are not interested in stable states of such iterations but in their possible unpredictable behaviors. By establishing a link between chaotic iterations and the notion of Devaney's topological chaos, we give conditions ensuring that these kind of algorithms produce topological chaos. This leads to algorithms that are highly unpredictable.
After presenting the theoretical foundations of our approach, we are interested in its practical aspects. We show how the theoretical algorithms give rise to computer programs that produce true topological chaos, then we propose applications in the area of information security.
\end{abstract}
%\underline{\textit{Key-words:}} Topological chaos; chaotic iterations; chaotic algorithms; Mealy machine; secure data hiding; hash functions.
\section{INTRODUCTION}
The use of chaos in various fields of information security such as data hiding, hash functions, or pseudo-random number generators is almost always based on the conception of algorithms that include known chaotic maps such as the logistic map. The goal is to obtain an algorithm which preserves the chaotic properties of the included chaotic functions.
For example, in \cite{Wu2007} and \cite{Wu2007bis}, a watermark $W$ is encrypted in $W_{e}$ by using the bitwise exclusive or: $W_{e}=W\otimes X$, where $X$ is a logistic map. Then, pixels of the carrier image designed to embed these bits are selected with the 2-D Arnold's cat map. A similar use of chaotic maps for watermarking can be found in e.g. \cite{Zhao2004}, \cite{Zhou1997429}, \cite{liu2007bis}, \cite{Cong2006} and \cite{Dawei2004}. In the domain of hash functions, the use of chaotic maps is seen in e.g. \cite{Fei2005}, \cite{Wang2003}, \cite{Xiao20094346} and \cite{Peng2005}.
However, without rigorous proof it is not indisputable that an algorithm that includes chaotic functions preserves chaos properties: for example, using the logistic function with other ``obvious'' parameters does not guarantee that the algorithm is chaotic.
Moreover, even if the algorithm obtained by the inclusion of chaotic maps is itself chaotic, the implementation of this algorithm on a machine can cause it to lose its chaotic nature. This is due to the finite nature of the machine numbers set. These issues are discussed in this document.

In this paper we don't simply integrate chaotic maps hoping that the security algorithm remains chaotic, we conceive algorithms for computer security that we mathematically prove chaotic. We raise the question of their implementation, proving in doing so that it is possible to design a chaotic algorithm and a chaotic computer program.

The chaos theory we consider is Devaney's topological chaos. In addition to being recognized as one of the best mathematical definition of chaos, this theory offers a framework with qualitative and quantitative tools to evaluate the notion of unpredictability. As an application of our fundamental results, we are interested in the area of information security. We propose in this paper a new approach of security which is based on unpredictability as it is defined by Devaney's chaos.

The paper begins by introducing the theoretical foundation of the new approach. We recall the definition of Devaney's topological chaos as well as the definition of discrete chaotic iterations. Although these definitions are distinct from each other, we establish a link between them by giving conditions under which chaotic discrete iterations generate a Devaney's topological chaos. Because chaotic iterations are very suited for computer programming, this link allows us to generate Devaney's chaos topological in the computer science field.
After having studied the theoretical aspects of our approach we focus on practical aspects. The important question is how to preserve the topological chaos properties in a set of a finite number of states. This question is answered by introducing a concept we call \emph{secure chaotic information machine}. This is a Mealy machine generating chaos as defined by Devaney (Section \ref{section:CHAOS IN A FINITE STATE MACHINE}).
We also give some applications of our approach of chaos, in the domain of information security. Algorithms intended for information security and based on this new approach are explained in Section \ref{section:APPLICATIONS IN COMPUTER SCIENCE}, in the hash function domain
\medskip

The rest of this paper is organized as follows. In Section \ref{section:BASIC RECALLS}, the definitions of Devaney's chaos and discrete chaotic iterations are recalled. A link between these two notions is established and sufficient conditions to obtain Devaney's topological chaos from discrete chaotic iterations are given in Section \ref{section:CHAOTIC ITERATIONS AS DEVANEY'S CHAOS}. In Section \ref{section:CHAOS IN A FINITE STATE MACHINE}, the question on how to apply the theoretical result is raised and applications in the computer science field are given in Section \ref{section:APPLICATIONS IN COMPUTER SCIENCE}. The paper ends with a conclusion in which our contribution is summarized and the planned future work is discussed.
%\section{RELATED WORKS AND CONTRIBUTION}
\label{section:RELATED WORKS AND CONTRIBUTION}
\section{BASIC RECALLS}
\label{section:BASIC RECALLS}
This section is devoted to basic definitions and terminologies in the field of topological chaos and in the one of chaotic iterations.
\subsection{Devaney's chaotic dynamical systems}
Consider a metric space $(\mathcal{X},d)$ and a continuous function $f:\mathcal{X}\longrightarrow \mathcal{X}$.
\begin{definition}
$f$ is said to be \emph{topologically transitive} if, for any pair of open sets $U,V \subset \mathcal{X}$, there exists $k>0$ such that $f^k(U) \cap V \neq \varnothing$.
\end{definition}
\begin{definition}
An element (a point) $x$ is a \emph{periodic element} (point) for $f$ of period $n\in \mathds{N}^*,$ if $f^{n}(x)=x$. The set of periodic points of $f$ is denoted $Per(f).$
\end{definition}
\begin{definition}
$(\mathcal{X},f)$ is said to be \emph{regular} if the set of periodic points is dense in $\mathcal{X}$,
\begin{equation*}
\forall x\in \mathcal{X},\forall \varepsilon >0,\exists p\in Per(f),d(x,p)\leqslant \varepsilon .
\end{equation*}
\end{definition}
\begin{definition}
\label{sensitivity} $f$ has \emph{sensitive dependence on initial conditions}
if there exists $\delta >0$ such that, for any $x\in \mathcal{X}$ and any neighborhood $V$ of $x$, there exists $y\in V$ and $n\geqslant 0$ such that $|f^{n}(x)-f^{n}(y)|>\delta $.
$\delta$ is called the \emph{constant of sensitivity} of $f$.
\end{definition}
Let us now recall the definition of a chaotic topological system, in the
sense of Devaney~\cite{Devaney}:
\begin{definition}
$f:\mathcal{X}\longrightarrow \mathcal{X}$ is said to be \emph{chaotic} on $%
\mathcal{X}$ if,
\begin{enumerate}
\item $f$ has sensitive dependence on initial conditions,
\item $f$ is topologically transitive,
\item $(\mathcal{X},f)$ is regular.
\end{enumerate}
\end{definition}
Therefore, quoting Robert Devaney: ``A chaotic map possesses three ingredients: unpredictability, indecomposability and an element of regularity. A chaotic system is unpredictable because of the sensitive dependence on initial conditions. It cannot be broken down or decomposed into two subsystems, because of topological transitivity. And, in the midst of this random behavior, we nevertheless have an element of regularity, namely the periodic points which are dense.'' Fundamentally different behaviors are thus possible and occur in an unpredictable way.
\subsection{Chaotic iterations}
\label{sec:chaotic iterations}
In the sequel $S^{n}$ denotes the $n^{th}$ term of a sequence $S$, $V_{i}$
denotes the $i^{th}$ component of a vector $V$ and $f^{k}=f\circ ...\circ f$
denotes the $k^{th}$ composition of a function $f$. Finally, the following
notation is used: $\llbracket1;N\rrbracket=\{1,2,\hdots,N\}$.
Let us consider a \emph{system} of a finite number $\mathsf{N}$ of elements (or \emph{%
cells}), so that each cell has a boolean \emph{state}. Then a sequence of
length $\mathsf{N}$ of boolean states of the cells corresponds to a
particular \emph{state of the system}. A sequence which elements belong to $ \llbracket 1;\mathsf{N} \rrbracket $ is called a \emph{strategy}. The set of all strategies is denoted by $\mathbb{S}.$
\begin{definition}
\label{Def:chaotic iterations}
The set $\mathds{B}$ denoting $\{0,1\}$, let $f:\mathds{B}^{\mathsf{N}%
}\longrightarrow \mathds{B}^{\mathsf{N}}$ be a function and $S\in \mathbb{S}
$ be a strategy. Then, the so-called \emph{chaotic iterations} are defined by $x^0\in \mathds{B}^{\mathsf{N}}$ and $\forall n\in \mathds{N}^{\ast },$
\begin{equation}
\forall i\in \llbracket1;\mathsf{N}\rrbracket%
,x_i^n=\left\{
\begin{array}{ll}
x_i^{n-1} & \text{ if }S^n\neq i \\
\left(f(x^{n-1})\right)_{S^n} & \text{ if }S^n=i.%
\end{array}%
\right.%
\end{equation}
\end{definition}
In other words, at the $n^{th}$ iteration, only the $S^{n}-$th cell is \textquotedblleft iterated\textquotedblright . Note that in a more general formulation, $S^n$ can be a subset of components and $f(x^{n-1})_{S^{n}}$ can be replaced by $f(x^{k})_{S^{n}}$ (where $k\leqslant n-1$), describing for example delays transmission (see e.g.~\cite{Bahi2002}). For the general definition of such chaotic iterations, see e.g.~\cite{Robert1986}.
\section{CHAOTIC ITERATIONS AS DEVANEY'S CHAOS}
\label{section:CHAOTIC ITERATIONS AS DEVANEY'S CHAOS}
\subsection{The new topological space}
In this section we will put our study in a topological context by defining a suitable metric space where chaotic iterations are continuous.
\subsubsection{Defining the iteration function and the phase space}
\label{Defining}
Let $\delta $ be the \emph{discrete boolean metric}, $\delta (x,y)=0\Leftrightarrow x=y.$ Given a function $f$, define the function:
\begin{equation*}
\begin{array}{lrll}
F_{f}: & \llbracket1;\mathsf{N}\rrbracket\times \mathds{B}^{\mathsf{N}} &
\longrightarrow & \mathds{B}^{\mathsf{N}} \\
& (k,E) & \longmapsto & \left( E_{j}.\delta (k,j)+f(E)_{k}.\overline{\delta
(k,j)}\right) _{j\in \llbracket1;\mathsf{N}\rrbracket},%
\end{array}%
\end{equation*}%
\noindent where + and . are the boolean addition and product operations.
Consider the phase space:
\begin{equation*}
\mathcal{X} = \llbracket 1 ; \mathsf{N} \rrbracket^\mathds{N} \times
\mathds{B}^\mathsf{N},
\end{equation*}
\noindent and the map defined on $\mathcal{X}$:
\begin{equation}
G_f\left(S,E\right) = \left(\sigma(S), F_f(i(S),E)\right), \label{Gf}
\end{equation}
\noindent where $\sigma$ is the \emph{shift} function defined by $\sigma (S^{n})_{n\in \mathds{N}}\in \mathbb{S}\longrightarrow (S^{n+1})_{n\in \mathds{N}}\in \mathbb{S}$ and $i$ is the \emph{initial function} \linebreak $i:(S^{n})_{n\in \mathds{N}} \in \mathbb{S}\longrightarrow S^{0}\in \llbracket 1;\mathsf{N}\rrbracket$. Then the chaotic iterations defined in (\ref{sec:chaotic iterations}) can be described by the following iterations:
\begin{equation*}
\left\{
\begin{array}{l}
X^0 \in \mathcal{X} \\
X^{k+1}=G_{f}(X^k).%
\end{array}%
\right.
\end{equation*}%
With this formulation, a shift function appears as a component of chaotic iterations. The shift function is a famous example of a chaotic map~\cite{Devaney} but its presence is not sufficient enough to claim $G_f$ as chaotic. In the rest of this section we prove rigorously that under some hypotheses, chaotic iterations generate topological chaos. Furthermore, due to the suitability of chaotic iterations for computer programming we also prove that this is true in the computer science field.
\subsubsection{Cardinality of $\mathcal{X}$}
By comparing $\mathbb{S}$ and $\mathds{R}$, we have the result.
\begin{theorem}
The phase space $\mathcal{X}$ has, at least, the cardinality of the continuum.
\end{theorem}
\begin{proof}
Let $\varphi$ be the map which transforms a strategy into the binary representation of an element in $[0,1$[, as follows. If the $n^{th}$ term of the strategy is 0, then the $n^{th}$ associated digit is 0, or else it is equal to 1.
With this construction, $\varphi : \llbracket 1 ; \mathsf{N} \rrbracket^\mathds{N} \longrightarrow [0,1]$ is surjective. But $]0,1[$ is isomorphic to $\mathds{R}$ ($x \in ]0,1[\mapsto tan(\pi(x-\frac{1}{2}))$ is an isomorphism), so the cardinality of $\llbracket 1 ; \mathsf{N} \rrbracket^\mathds{N}$ is greater or equal to the cardinality of $\mathds{R}$. As a consequence, the cardinality of the Cartesian product $\mathcal{X} = \llbracket 1 ; \mathsf{N} \rrbracket^\mathds{N} \times \mathds{B}^\mathsf{N}$ is greater or equal to the cardinality of $\mathds{R}$.
\end{proof}
\begin{remark}
This result is independent from the number of cells of the system.
\end{remark}
\subsubsection{A new distance}
We define a new distance between two points $X = (S,E), Y = (\check{S},\check{E})\in
\mathcal{X}$ by%
\begin{equation*}
d(X,Y)=d_{e}(E,\check{E})+d_{s}(S,\check{S}),
\end{equation*}
\noindent where
\begin{equation*}
\left\{
\begin{array}{lll}
\displaystyle{d_{e}(E,\check{E})} & = & \displaystyle{\sum_{k=1}^{\mathsf{N}%
}\delta (E_{k},\check{E}_{k})}, \\
\displaystyle{d_{s}(S,\check{S})} & = & \displaystyle{\dfrac{9}{\mathsf{N}}%
\sum_{k=1}^{\infty }\dfrac{|S^k-\check{S}^k|}{10^{k}}}.%
\end{array}%
\right.
\end{equation*}
If the floor value $\lfloor d(X,Y)\rfloor $ is equal to $n$,
then the systems $E, \check{E}$ differ in $n$ cells. In addition, $d(X,Y) - \lfloor d(X,Y) \rfloor $ is a measure of the differences between strategies $S$ and $\check{S}$. More precisely, this floating part is less than $10^{-k}$ if and only if the first $k$
terms of the two strategies are equal. Moreover, if the $k^{th}$ digit is nonzero, then the $k^{th}$ terms of the two
strategies are different.
\subsubsection{Continuity of the iteration function}
To prove that chaotic iterations are an example of topological chaos in the
sense of Devaney ~\cite{Devaney}, $G_{f}$ must be continuous in the metric
space $(\mathcal{X},d)$.
\begin{theorem}
$G_f$ is a continuous function.
\end{theorem}
\begin{proof}
We use the sequential continuity.
Let $(S^n,E^n)_{n\in \mathds{N}}$ be a sequence of the phase space $%
\mathcal{X}$, which converges to $(S,E)$. We will prove that $\left(
G_{f}(S^n,E^n)\right) _{n\in \mathds{N}}$ converges to $\left(
G_{f}(S,E)\right) $. Let us recall that for all $n$, $S^n$ is a strategy,
thus, we consider a sequence of strategies (\emph{i.e.} a sequence of
sequences).\newline
As $d((S^n,E^n);(S,E))$ converges to 0, each distance $d_{e}(E^n,E)$ and $d_{s}(S^n,S)$ converges
to 0. But $d_{e}(E^n,E)$ is an integer, so $\exists n_{0}\in \mathds{N},$ $%
d_{e}(E^n,E)=0$ for any $n\geqslant n_{0}$.\newline
In other words, there exists a threshold $n_{0}\in \mathds{N}$ after which no
cell will change its state:
\begin{equation*}
\exists n_{0}\in \mathds{N},n\geqslant n_{0}\Longrightarrow E^n = E.
\end{equation*}%
In addition, $d_{s}(S^n,S)\longrightarrow 0,$ so $\exists n_{1}\in %
\mathds{N},d_{s}(S^n,S)<10^{-1}$ for all indexes greater than or equal to $%
n_{1}$. This means that for $n\geqslant n_{1}$, all the $S^n$ have the same
first term, which is $S^0$:%
\begin{equation*}
\forall n\geqslant n_{1},S_0^n=S_0.
\end{equation*}%
Thus, after the $max(n_{0},n_{1})^{th}$ term, states of $E^n$ and $E$ are
identical and strategies $S^n$ and $S$ start with the same first term.\newline
Consequently, states of $G_{f}(S^n,E^n)$ and $G_{f}(S,E)$ are equal,
so, after the $max(n_0, n_1)^{th}$ term, the distance $d$ between these two points is strictly less than 1.\newline
\noindent We now prove that the distance between $\left(
G_{f}(S^n,E^n)\right) $ and $\left( G_{f}(S,E)\right) $ is convergent to
0. Let $\varepsilon >0$. \medskip
\begin{itemize}
\item If $\varepsilon \geqslant 1$, we see that distance
between $\left( G_{f}(S^n,E^n)\right) $ and $\left( G_{f}(S,E)\right) $ is
strictly less than 1 after the $max(n_{0},n_{1})^{th}$ term (same state).
\medskip
\item If $\varepsilon <1$, then $\exists k\in \mathds{N},10^{-k}\geqslant
\varepsilon \geqslant 10^{-(k+1)}$. But $d_{s}(S^n,S)$ converges to 0, so
\begin{equation*}
\exists n_{2}\in \mathds{N},\forall n\geqslant
n_{2},d_{s}(S^n,S)<10^{-(k+2)},
\end{equation*}%
thus after $n_{2}$, the $k+2$ first terms of $S^n$ and $S$ are equal.
\end{itemize}
\noindent As a consequence, the $k+1$ first entries of the strategies of $%
G_{f}(S^n,E^n)$ and $G_{f}(S,E)$ are the same ($G_{f}$ is a shift of strategies) and due to the definition of $d_{s}$, the floating part of
the distance between $(S^n,E^n)$ and $(S,E)$ is strictly less than $%
10^{-(k+1)}\leqslant \varepsilon $.\bigskip \newline
In conclusion,
$$
\forall \varepsilon >0,\exists N_{0}=max(n_{0},n_{1},n_{2})\in \mathds{N}%
,\forall n\geqslant N_{0},
$$
$$
 d\left( G_{f}(S^n,E^n);G_{f}(S,E)\right)
\leqslant \varepsilon .
$$
$G_{f}$ is consequently continuous.
\end{proof}
In this section, we proved that chaotic iterations can be modeled as a dynamical system in a topological space. In the next section, we show that chaotic iterations are a case of topological chaos, according to Devaney.
% >>>>>>>>>>>>>>>>>>>>>> Discrete chaotic iterations as topological chaos <<<<<<<<<<<<<<<<<<<<<<<<<<<<<<
\subsection{Discrete chaotic iterations as topological chaos}
To prove that we are in the framework of Devaney's topological chaos, we have to find a boolean function $f$ such that $G_f$ satisfies the regularity, transitivity and sensitivity conditions. We will prove that the vectorial logical negation
\begin{equation}
f_{0}(x_{1},%
\hdots,x_{\mathsf{N}})=(\overline{x_{1}},\hdots,\overline{x_{\mathsf{N}}})
\label{f0}
\end{equation}
\noindent is a suitable function.
\subsubsection{Regularity}
\label{regularite}
\begin{theorem}
Periodic points of $G_{f_0}$ are dense in $\mathcal{X}$.
\end{theorem}
\begin{proof}
Let $(\check{S}, \check{E})\in \mathcal{X}$ and $\varepsilon >0$. We are
looking for a periodic point $(\widetilde{S},\widetilde{E})$ satisfying $d((%
\check{S}, \check{E});(\widetilde{S},\widetilde{E}))<\varepsilon$.
As $\varepsilon$ can be strictly lesser than 1, we must choose $%
\widetilde{E} = \check{E}$. Let us define $k_0(\varepsilon) =\lfloor
log_{10}(\varepsilon )\rfloor +1$ and consider the set
\[
\mathcal{S}_{\check{S}, k_0(\varepsilon)} = \left\{ S \in \mathbb{S} / S^k =
\check{S}^k, \forall k \leqslant k_0(\varepsilon) \right\}.
\]
Then, $\forall S \in \mathcal{S}_{\check{S}, k_0(\varepsilon)}, d((S, \check{%
E});(\check{S}, \check{E})) < \varepsilon$. It remains to choose $\widetilde{%
S} \in \mathcal{S}_{\check{S}, k_0(\varepsilon)}$ such that $(\widetilde{S},%
\widetilde{E}) = (\widetilde{S},\check{E})$ is a periodic point for $%
G_{f_0}$.
Let $\mathcal{J} = \left\{ i \in \{1, ..., \mathsf{N}\} / E_i \neq \check{%
E}_i, \text{ where } (S, E) = G_{f_0}^{k_0} (\check{S}, \check{E}) \right\}$%
, $i_0 = card(\mathcal{J})$ and $j_1 <j_2 < ... < j_{i_0}$ the elements of $%
\mathcal{J}$. Then, $\widetilde{S} \in \mathcal{S}_{\check{S},
k_0(\varepsilon)}$ defined by
\begin{itemize}
\item $\widetilde{S}^k = \check{S}^k$, if $k \leqslant k_0(\varepsilon)$,
\item $\widetilde{S}^k = j_{k-k_0(\varepsilon)}$, if $k \in
\{k_0(\varepsilon)+1, k_0(\varepsilon)+2, ..., k_0(\varepsilon)+i_0\}$,
\item and $\widetilde{S}^{k}=\widetilde{S}^{j}$, where $j\leqslant
k_{0}(\varepsilon )+i_{0}$ is satisfying $j\equiv k~(\text{mod }%
k_{0}(\varepsilon )+i_{0})$, if $k>k_{0}(\varepsilon )+i_{0}$,
\end{itemize}
\noindent is such that $%
(\widetilde{S},\widetilde{E})$ is a periodic point, of period $%
k_{0}(\varepsilon )+i_{0}$, which is $\varepsilon -$closed to $(\check{S},%
\check{E})$.\newline As a conclusion, $(\mathcal{X},G_{f_{0}})$ is
regular.
\end{proof}
\subsubsection{Transitivity}
\begin{theorem}
$(\mathcal{X},G_{f_0})$ is topologically transitive.
\end{theorem}
\begin{proof}
Let us define $\mathcal{E}:\mathcal{X}\rightarrow \mathbb{B}^{\mathsf{N}},$
such that $\mathcal{E(}S,E)=E.$ Let $\mathcal{B}_{A}=\mathcal{B}(X_{A},r_{A})
$ and $\mathcal{B}_{B}=\mathcal{B}(X_{B},r_{B})$ be two open balls of $%
\mathcal{X}$, with $X_{A}=(S_{A},E_{A})$ and $X_{B}=(S_{B},E_{B})$. We are
looking for $\widetilde{X}=(\widetilde{S},\widetilde{E})$ in $\mathcal{B}_{A}
$ such that $\exists n_{0}\in \mathbb{N},G_{f_{0}}^{n_{0}}(\widetilde{X})\in
\mathcal{B}_{B}$.\newline
$\widetilde{X}$ must be in $\mathcal{B}_{A}$ and $r_{A}$ can be strictly
lesser than 1, so $\widetilde{E}=E_{A}$. Let $k_{0}=\lfloor \log
_{10}(r_{A})+1\rfloor $. Then $\forall S\in \mathbb{S}$, if $%
S^{k}=S_{A}^{k},\forall k\leqslant k_{0}$, then $(S,\widetilde{E})\in
\mathcal{B}_{A}$. Let us notice $(\check{S},\check{E}%
)=G_{f_{0}}^{k_{0}}(S_{A},E_{A})$ and $c_{1},...,c_{k_{1}}$ the elements of
the set $\{i\in \llbracket1,\mathsf{N}\rrbracket/\check{E}_{i}\neq \mathcal{E%
}(X_{B})_{i}\}.$ So any point $X$ of the set
\[
\{(S,E_{A})\in \mathcal{X}/\forall k\leqslant k_{0},S^{k}=S_{A}^{k}\text{
and }\forall k\in \llbracket1,k_{1}\rrbracket,S^{k_{0}+k}=c_{k}\}
\]%
is satisfying $X\in \mathcal{B}_{A}$ and $\mathcal{E}\left(
G_{f_{0}}^{k_{0}+k_{1}}(X)\right) =E_{B}$.
\noindent Lastly, let us define $k_2 = \lfloor \log_{10}(r_B)\rfloor +1$. Then $%
\widetilde{X} = (\widetilde{S}, \widetilde{E}) \in \mathcal{X}$ defined by:
\begin{enumerate}
\item $\widetilde{X} = E_A$,
\item $\forall k \leqslant k_0, \widetilde{S}^k = S_A^k$,
\item $\forall k \in \llbracket 1, k_1 \rrbracket,$ $\widetilde{S}^{k_0+k} =
c_k$,
\item $\forall k \in \mathbb{N}^*, \widetilde{S}^{k_0+k_1+k} = S_B^k$,
\end{enumerate}
\noindent is such that $\widetilde{X} \in \mathcal{B}_A$ and $G_{f_0}^{k_0+k_1}(%
\widetilde{X}) \in \mathcal{B}_B$.
\end{proof}

\subsubsection{Sensitive dependence on initial conditions}
\begin{theorem}
$(\mathcal{X},G_{f_0})$ has sensitive dependence on initial conditions.
\end{theorem}
\begin{proof}
Banks \emph{et al.} proved in ~\cite{Banks92} that having sensitive dependence is a consequence of being regular and topologically transitive.
\end{proof}

\subsubsection{Devaney's Chaos}
\label{sec:DevaneysChaos}
In conclusion, $(\mathcal{X},G_{f_0})$ is topologically transitive, regular and has sensitive dependence on initial conditions. Then we have the following result:
\begin{theorem}
\label{theorem:Chaos}
$G_{f_0}$ is a chaotic map on $(\mathcal{X},d)$ in the sense of Devaney.
\end{theorem}
\begin{remark}
%We have proved that the set of the iterate functions $f$ such that $(\mathcal{X}, G_f)$ is chaotic (in the meaning of Devaney), is a nonempty set. In a future work, we will give a characterization of this set.
We have proven that the set $\mathcal{C}$ of the iterate functions $f$ so that $(\mathcal{X}, G_f)$ is chaotic (according to the definition of Devaney), is a nonempty set. In future work, we will deepen the study of $\mathcal{C}$, among other things, by computing its cardinality and characterizing this set.
\end{remark}
\section{CHAOS IN A FINITE STATE MACHINE}
\label{section:CHAOS IN A FINITE STATE MACHINE}
\subsection{The approach presented in this paper}
In the section above, it has been proven that discrete chaotic iterations can be put in the field of discrete dynamical systems:
\begin{equation}
\left\{
\begin{array}{l}
x^{0}\in \mathcal{X} \\
x^{n+1} = G_f(x^{n}),
\end{array}%
\right.
\end{equation}
where $(\mathcal{X},d)$ is a metric space and $G_f$ is a continuous function. Thus, it becomes possible to study the topological behavior of those chaotic iterations.
Precisely, it has been proven that if the iterate function is based on the vectorial logical negation $f_0$, then chaotic iterations generate chaos according to Devaney. Therefore chaotic iterations, as Devaney's topological chaos, satisfy: sensitive dependence on the initial conditions, unpredictability, indecomposability, and uniform repartition.
\medskip

Two major problems typically occur when trying to develop a computer program with chaotic behavior. First, computers have a finite number of states, so the computations always enter into cycles. Second, the properties of chaotic algorithms are inherited from a real chaotic sequence (like a logistic map) and this behavior is lost when computing floating-point numbers (unlike real numbers, floating-point numbers have a finite decimal part). These two problems are solved in this paper due to the two following ideas:

\begin{enumerate}
\item Chaotic iterations are Mealy machines. %Chaotic iterations are modeled as a Mealy machine.
% ATTENTION : Ne pas parler de condition initiale. Moore = Condition initiale unique, puis itérations, enfin résultat. Mealy = spécificité des itérations chaotiques : stratégie variant au cours du temps. Ceci est modélisé par une machine de MEaly, i.e. stratégie prise du monde extérieur. A chaque itération, mémoire et stratégie sont mises à jour
%At each iteration, data corresponding to the current strategy \ref{sec:chaotic iterations} are taken from the outside world, then computations are realized into the memory (the updates of the finite state of the system) and the last state is returned after a desired number of iterations. Contrary to the existing points of view, based on a Moore machine, this machine can pass two times in a same state, without continuing with the same evolution. Section \ref{subsection: A chaotic Mealy machine} explains in details this original contribution, which allows the realization of a true chaos in computers.
At each iteration, data corresponding to the current strategy \ref{sec:chaotic iterations} are taken from the outside world, then computations are realized into the memory (the updates of the finite state of the system). The last state is returned after a desired number of iterations. Contrary to the existing points of view, based on a Moore machine, this machine can pass two times in a same state, without continuing the same evolution. Section \ref{subsection: A chaotic Mealy machine} explains in detail this original contribution, which allows the realization of true chaos in computers.
\item As mentioned above, the strategy $S$ defined in \ref{sec:chaotic iterations} will not depend on real numbers, but on integers taken from the outside world. We work with the set $\mathcal{X}$ defined in Subsection \ref{Defining} which has the cardinality of the continuum. Section \ref{subsection:The particular case of regularity} discusses the consequences of dealing with finite strategies in practice.
% ATTENTION : Nous ne travaillons pas avec l'ensemble fini des nombres machines, (quel sens donner au chaos, si on a un nombre fini de possibilités?). Mais avec les éléments de X, qui peuvent être mis en bijection avec l'ensemble des décimaux, infinis dense dans R, ceci grace à S.
\end{enumerate}
%A chaotic behavior can serve to increase confidence into programs used in computer science security fields. Section \ref{section:APPLICATIONS IN COMPUTER SCIENCE} shows that chaotic iterations can be used to realize such chaotic and secure programs. A data hiding chaotic algorithm, following by a chaotic hash function, are given as illustration examples. Some qualitative and quantitative evaluations of the security of these algorithms, in topological terms, are given finally in Section \ref{section:PROPERTIES OF THE CHAOTIC MACHINE}.
\medskip
%Discrete chaotic iterations represent a particular class of iterative algorithms: some data are taken from the outside world (the initial conditions), then computations are realized into the memory (the updates of the finite state of the system) and the last state is returned after a desired number of iterations. The theorem ~\ref{theorem:Chaos} gives the proof that the algorithm is chaotic in the meaning of Devaney, at least when using $G_{f_0}$. In particular, the outputs of this algorithm will be highly sensitive to the data inputed in it. It remains to be certain that chaos is not lost after the computation of this algorithm and to discover how to use it in concrete applications.
\medskip
\subsection{A chaotic Mealy machine}
\label{subsection: A chaotic Mealy machine}
%Some computer security fields like cryptography, pseudorandom number generation or data hiding, can have an interest to work into a rigorous chaos framework. This mathematical theory brings some qualitative and quantitative tools allowing the proof and evaluation of qualities required in these fields: unpredictability, disorder, uniform repartition, sensitivity, topological entropy... The security of existing algorithms could then be reinforced, while new chaotic methods can be discovered. However, until now, most methods that have been presented lose theirs chaotic properties when computing in the finite set of machine numbers.
The algorithms considered chaotic usually follow the principle of a Moore machine. After having received its initial states, the machine works alone with no interaction with the outside world. Its outputs only depend on the different states of the machine. The main problem is that when a machine with a finite number of states reaches a same state twice, the two following evolutions are identical. Such an algorithm always enters into a cycle. This behavior is highly predictable and cannot be set as chaotic. Attempts to define a discrete notion of chaos have been proposed, but they are not completely satisfactory and are less recognized than the notion of Devaney's topological chaos.
This problem does not occur in a Mealy machine. This finite state transducer generates an output $O$ computed from its current state $E$ \emph{and} the current value of an input $S$ (Fig. \ref{fig:Mealy}). %ATTENTION : Développer.
By this accord, even if the machine reaches the same state twice, the corresponding following evolutions may be completely different depending on the values of the inputs. The method presented here is based on such a machine. Indeed, chaotic iterations are a Mealy machine: at each iteration, the computations take into account new inputs (strategies) which are obtained, for example, from the media on which our algorithm applies. Roughly speaking, as the set of all media is infinite, we obtain a finite state machine which can evolve in infinite ways, thus making it possible to obtain a true chaos in computers.
\begin{figure}[htb]
\centerline{\epsfig{figure=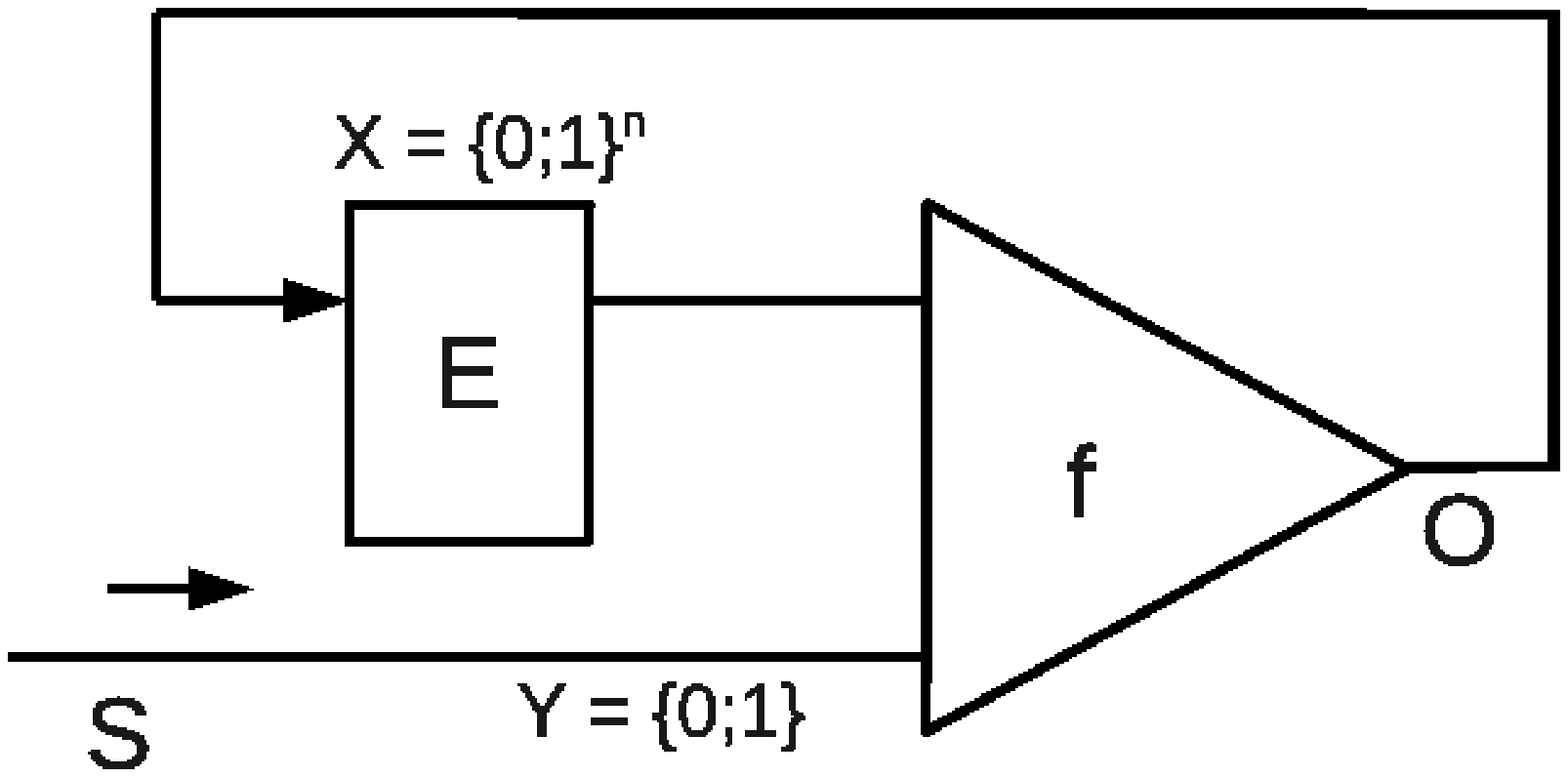,width=5.cm}}
\caption{Mealy machine for chaotic algorithms.}
\label{fig:Mealy}
\end{figure}
%Following these conclusions, we have searched some applications in computer science which operate on media and have an interest to behave as chaotic. Cryptography, hash functions, data hiding and steganography offer applications that satisfy these requirements. We have looking for algorithms involved in these applications, which can be modeled as a Mealy machine. Its new state will depend on a new input at each update. The next value of this input can be calculated from the next bits read from the media. Last, this Mealy machine must satisfy the three properties of Devaney: the proposed algorithm will thus be chaotic, in the meaning of Devaney.
\begin{definition}
A Mealy machine is said to be chaotic if this machine has a chaotic behavior, as expressed by Devaney.
\end{definition}
The Mealy machine we used in this document will be the chaotic iterations with $G_{f_0}$ as iterate function. Because these chaotic iterations satisfy the Devaney's definition of chaos, as stated in Section \ref{sec:DevaneysChaos}, we can conclude that our Mealy machine is a chaotic machine.
% ATTENTION : De la section suivante, faire:
% - une remarque, à placer ailleurs,
% - le reste, en anne
\subsection{The practical case of finite strategies}
\label{subsection:The particular case of regularity}
It is worthwhile to notice that even if the set of machine numbers is finite, we deal in practice with the \emph{infinite} set of strategies that have finite but unbounded lengths. Indeed, as suggested before, it is not necessary to store all the terms of the strategy in the memory. Only its $n^{th}$ term (an integer less than or equal to $\mathsf{N}$) has to be stored at the $n^{th}$ step, as it is illustrated in the following example. Let us suppose that a given text is input from the outside world into the computer character by character and that the current term of the strategy is computed from the ASCII code of the current stored character. Since the set of all possible texts of the outside world is infinite and the number of their characters is unbounded, we work with an infinite set of finite but unbounded strategies. Of course, the previous example is a simplistic one. A chaotic procedure should to be introduced to generate the terms of the strategy from the stream of characters.\newline
In the computer science framework, we also have to deal with a finite set of states of the form $\mathds{B}^\mathsf{N}$ and as stated before an infinite set $\mathbb{S}$ of strategies. The sole difference with the theoretical study is that instead of being infinite the sequences of $S$ are finite with unbounded length.\newline
The proofs of continuity and transitivity are independent of the finiteness of the length of strategies (sequences of $\mathbb{S}$). Sensitivity can be proved too in this situation. So even in the case of finite machine numbers, we have the two fundamental properties of chaos: sensitivity and transitivity, which respectively implies unpredictability and indecomposability (see~\cite{Devaney}, p.50). The regularity supposes that the sequences are of infinite lengths. To obtain the analogous of regularity in the context of finite sets, we define below the notion of \emph{periodic but finite} sequences.
\begin{definition}
A strategy $S\in\mathbb{S}$ is said to be \emph{periodic but finite} if $S$ is a finite sequence of length $n$ and if there exists a divisor $p$ of $n$, $p \neq n$, such that $\forall i \leqslant n-p, S^i = S^{i+p}$. A point $(E,S) \in \mathcal{X}$ is said to be \emph{periodic but finite}, if its strategy $S$ is periodic but finite.
\end{definition}
For example, $(1,2,1,2,1,2,1,2)$ ($p$=2) and $(2,2,2)$ ($p$=1), are periodic but finite. This definition can be interpreted as the analogous of periodicity definition on finite strategies. Following the proof of regularity (Section \ref{regularite}), it can be proven that the set of periodic but finite points is dense on $\mathcal{X}$, hence obtaining a desired element of regularity in finite sets, as quoted by Devaney (\cite{Devaney}, p.50): ``two points arbitrary close to each other could have completely different behaviors, the one could have a cyclic behavior as long as the system iterates while the trajectory of the second could `visit' the whole phase space''. It should be recalled that the regularity was introduced by Devaney in order to counteract the effects of sensitivity and transitivity: two points close to each other can have fundamentally different behaviors.
\section{HASH FUNCTIONS BASED ON TOPOLOGICAL CHAOS}
\label{section:APPLICATIONS IN COMPUTER SCIENCE}
\subsection{Introduction}
The use of chaotic maps to generate hash algorithms has seen several developments in recent years. In \cite{Fei2005} for example, a digital signature algorithm based on an elliptic curve and chaotic mapping is proposed to strengthen the security of an elliptic curve digital signature algorithm. Other examples of the generation of a hash function using chaotic maps can be found in \emph{e.g.} \cite{Wang2003}, \cite{Xiao20094346} and \cite{Peng2005}. However, as for digital watermarking, the use of any chaotic map does not guarantee that the resulting hash function would behave chaotically too. To our knowledge, this point is not discussed in these referenced papers, however it should be considered as important.
We define in this section a new way to construct hash functions based on chaotic iterations. As a consequence of the theory presented before, the generated hash functions satisfy the topological chaos property. Thus, various desired properties in this domain are guaranted by our approach. For example, the avalanche criterion is closely linked to the sensitivity property. 
\subsection{A chaotic machine for hash functions}
In this section, we explain a new way to obtain a hash value of a digital medium described by a binary sequence. It is based on chaotic iterations and satisfies the topological chaos property. The hash value will be the last state of some chaotic iterations: the initial state $X_0$, finite strategy $S$ and iterate function must then be defined.
\label{subsubsec:initial}
\medskip
The initial condition $X_0=\left( S,E\right) $ is composed by a $\mathsf{N} = 256$ bits sequence $E$ and a chaotic strategy $S$. In the following sequence, we describe in detail how to obtain this initial condition from the original medium.

\subsubsection{How to obtain $E$}
The first step of our algorithm is to transform the message in a normalized 256 bits sequence $E$. To illustrate this step we state that our original text is: ``\emph{The original text}''.
Each character of this string is replaced by its ASCII code (on 7 bits). Then, we add a 1 to this string.
\bigskip
\begin{center}
\begin{alltt}
\noindent 10101001 10100011 00101010 00001101
\noindent 11111100 10110100 11100111 11010011
\noindent 10111011 00001110 11000100 00011101
\noindent 00110010 11111000 11101001
\end{alltt}
\end{center}
\bigskip
So, the binary value (1111000) of the length of this string (120) is added, with another 1:
\bigskip
\begin{center}
\begin{alltt}
\noindent 10101001 10100011 00101010 00001101
\noindent 11111100 10110100 11100111 11010011
\noindent 10111011 00001110 11000100 00011101
\noindent 00110010 11111000 11101001 11110001
\end{alltt}
\end{center}
\bigskip
The whole string is copied, but in the opposite direction. This gives:
\bigskip
\begin{center}
\begin{alltt}
\noindent 10101001 10100011 00101010 00001101
\noindent 11111100 10110100 11100111 11010011
\noindent 10111011 00001110 11000100 00011101
\noindent 00110010 11111000 11101001 11110001
\noindent 00011111 00101110 00111110 10011001
\noindent 01110000 01000110 11100001 10111011
\noindent 10010111 11001110 01011010 01111111
\noindent 01100000 10101001 10001011 0010101
\end{alltt}
\end{center}
\medskip So, we obtain a multiple of 512, by duplicating this string enough and truncating at the next multiple of 512. This string in which the whole original text is contained, is denoted by $D$.
\bigskip

Finally, we split our obtained string into blocks of 256 bits and apply the exclusive-or function, obtaining a 256 bits sequence.
\bigskip
\begin{alltt}
\noindent 11111010 11100101 01111110 00010110
\noindent 00000101 11011101 00101000 01110100
\noindent 11001101 00010011 01001100 00100111
\noindent 01010111 00001001 00111010 00010011
\noindent 00100001 01110010 01000011 10101011
\noindent 10010000 11001011 00100010 11001100
\noindent 10111000 01010010 11101110 10000001
\noindent 10100001 11111010 10011101 01111101
\end{alltt}
So, in the context of Subsection \ref{subsubsec:initial}, $\mathsf{N}=256$ and $E$ is the above obtained sequence of 256 bits.
\medskip

We now have the definitive length of our digest. Note that a lot of texts have the same string. This is not a problem because the strategy we will build depends on the whole text.
Let us now build the strategy $S$.

\subsubsection{How to choose $S$}
To obtain the strategy $S$, an intermediate sequence $(u^n)$ is constructed from $D$ as follows:
\begin{itemize}
\item $D$ is split into blocks of 8 bits. Then $u^n$ is the decimal value of the $n^{th}$ block.
\item A circular rotation of one bit to the left is applied to $D$ (the first bit of $D$ is put on the end of $D$). Then the new string is split into blocks of 8 bits another time. The decimal values of those blocks are added to $(u^n)$.
\item This operation is repeated again 6 times.
\end{itemize}
\bigskip

It is now possible to build the strategy $S$:
\begin{equation*}
S^0 = u^0,~~~
S^n=(u^n+2\times S^{n-1}+n) ~(mod ~256).
\end{equation*}%
\noindent $S$ will be highly dependent to the changes of the original text, because $\theta \longmapsto 2\theta ~(mod ~1)$ is known to be chaotic as defined by Devaney \cite{Devaney}.

\subsubsection{How to construct the digest}
To construct the digest, chaotic iterations are done with initial state $X^0$,
\begin{equation*}
\begin{array}{rccc}
f: & \llbracket1,256\rrbracket & \longrightarrow & \llbracket1,256\rrbracket
\\
& (E_1,\hdots,E_{256}) & \longmapsto & (\overline{E_1},\hdots,\overline{%
E_{256}}),%
\end{array}%
\end{equation*}%
\noindent as iterate function and $S$ for the chaotic strategy.
\bigskip

\noindent The result of those iterations is a 256 bits vector. Its components are taken 4 per 4 bits and translated into hexadecimal numbers, to obtain the hash value:
\medskip
\begin{alltt}
\noindent 63A88CB6AF0B18E3BE828F9BDA4596A6
\noindent A13DFE38440AB9557DA1C0C6B1EDBDBD
\end{alltt}
\bigskip

To compare, if instead of using the text \textquotedblleft \textit{The original text}\textquotedblright\ we took \textquotedblleft \textit{the original text}\textquotedblright , the hash function returns:
\medskip
\begin{alltt}
\noindent 33E0DFB5BB1D88C924D2AF80B14FF5A7
\noindent B1A3DEF9D0E831194BD814C8A3B948B3
\end{alltt}
\bigskip

In this paper, the generation of hash value is done with the vectorial boolean negation $f_{0} $ defined in eq. (\ref{f0}). Nevertheless, the procedure remains general and can be applied with any function $f$ such that $G_f$ is chaotic. In the following subsection, a complete example of the procedure is given.

\subsection{Application example}
Let us consider the two black and white images of size $64 \times 64$ in Fig. \ref{Hash of some black and white images}, in which the pixel in position (40,40) has been changed.
\begin{figure}[h]
\centering
\subfigure[Original image.]{\includegraphics[width=0.21\textwidth]{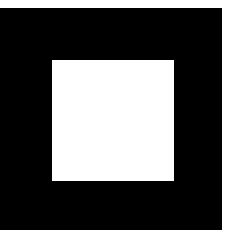}\label{Original image}}
\subfigure[Modified image.]{\includegraphics[width=0.21\textwidth]{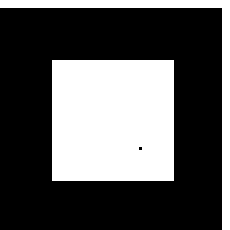}\label{Modified image}}
\caption{Hash of some black and white images.}
\label{Hash of some black and white images}
\end{figure}
\medskip
In this case, our hash function returns:
\begin{alltt}
\noindent 34A5C1B3DFFCC8902F7B248C3ABEFE2C
\noindent 9C9538E5104D117B399C999F74CF1CAD
\end{alltt}
for the Fig. \ref{Original image} and
\begin{alltt}
\noindent 5E67725CAA6B7B7434BE57F5F30F2D3D
\noindent 57056FA960B69052453CBC62D9267896
\end{alltt}
for the Fig. \ref{Modified image}.
\bigskip

Let us consider now the two 256 graylevel images of Lena ($256 \times 256$ pixels) in figure \ref{Hash of some grayscale level images}, in which the grayscale level of the pixel in position (50,50) has been transformed from 93 (fig. \ref{Original Lena}) to 94 (fig. \ref{Modified Lena}).
\begin{figure}[h]
\centering
\subfigure[Original lena.]{\includegraphics[width=0.18\textwidth]{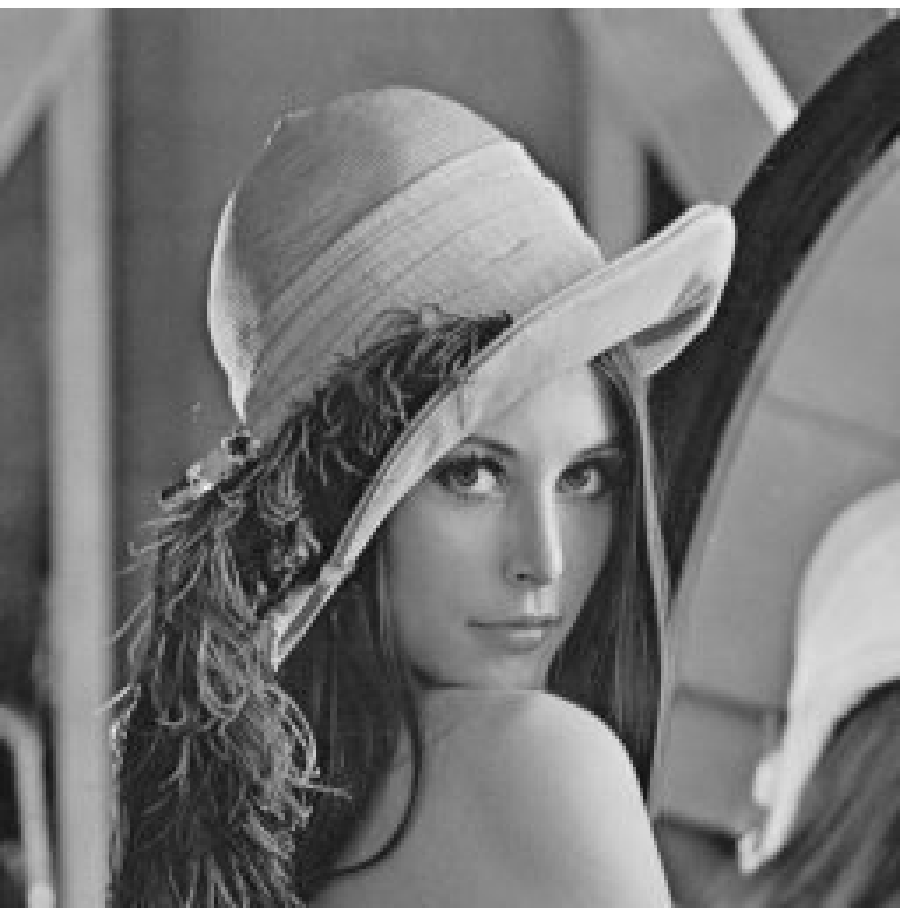}\label{Original Lena}}
\subfigure[Modified lena.]{\includegraphics[width=0.18\textwidth]{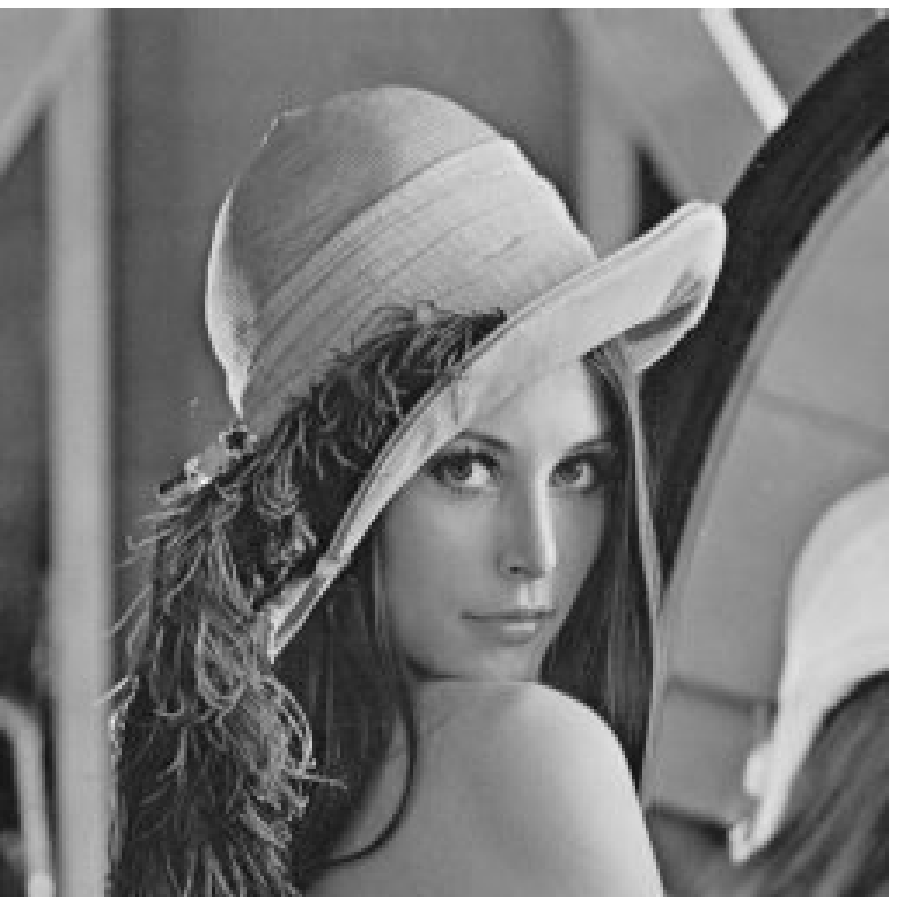}\label{Modified Lena}}
\caption{Hash of some grayscale level images.}
\label{Hash of some grayscale level images}
\end{figure}
In this case, our hash function returns:
\begin{alltt}
\noindent FA9F51EFA97808CE6BFF5F9F662DCD73
\noindent 8C25101FE9F7F427CD4E2B8D40331B89
\end{alltt}
for the left Lena and
\begin{alltt}
\noindent BABF2CE1455CA28F7BA20F52DFBD24B7
\noindent 6042DC572FCCA4351D264ACF4C2E108B
\end{alltt}
for the right Lena.
\medskip

These examples give an illustration of the avalanche effect obtained by this algorithm. A more complete study of the properties possessed by our hash functions and resistance under collisions will be studied in future work.
\section{CONCLUSION}
In this paper, a new approach to generate algorithms with chaotic behaviors is proposed. This approach is based on the well-known Devaney's topological chaos. The algorithms which are of iterative nature are based on the so-called chaotic iterations.
This is achieved by establishing a link between the notions of topological chaos and chaotic iterations. This is the first time that such an approach is considered for chaotic iterations. Indeed, we are not interested in stable states of such iterations as it has always been the case in the literature, but in their unpredictable behavior.
After a solid theoretical study, we consider the practical implementation of the proposed algorithms by evaluating the case of finite sets. We study the behavior of the induced computer programs proving that it is possible to design true chaotic computer programs.
A simple application is proposed in the area of hash functions. The security in this case is defined by the unpredictability of the behavior of the proposed algorithm.
The algorithm derived from our approach satisfies important properties of topological chaos such as sensitivity to initial conditions, uniform repartition (as a result of the transitivity), and unpredictability.
The results expected in our study have been experimentally checked. The choices made in this first study are simple: the aim was not to find the best hash function, but to give simple illustrated examples to prove the feasibility in using the new kind of chaotic algorithms in computer science.
In future work, we will investigate other choices of iteration functions and chaotic strategies. We will try to characterize transitive functions. Other properties induced by topological chaos, such as entropy, will be explored and their interest in the information security framework will be shown.
\bibliographystyle{plain}
\bibliography{wcci_papier4.bib}
\end{document}